\documentclass[aps,prd,twocolumn,superscriptaddress]{revtex4-2}

\usepackage{makecell} 
\usepackage{microtype}
\usepackage{graphicx}
\usepackage{dcolumn}
\usepackage{bm}
\usepackage{caption} 
\usepackage{mathtools} 
\usepackage{tensor} 
\usepackage[hidelinks]{hyperref} 
\usepackage{amssymb}
\usepackage{mathtools}
\usepackage{mathrsfs}
\usepackage{amsmath} 
\usepackage{xcolor}
\usepackage{array} 
\usepackage{booktabs}
\usepackage{subcaption}
\usepackage{enumitem}
\usepackage{soul}
\usepackage{amsthm}
\usepackage[toc,page]{appendix}

\newtheorem{theorem}{Theorem}

\begin{document}

\title{Generalized Ernst Potentials for arbitrary Dilatonic Theories}

\author{Leonel Bixano}
    \email{Contact author: leonel.delacruz@cinvestav.mx}
\author{Tonatiuh Matos}%
 \email{Contact author: tonatiuh.matos@cinvestav.mx}
\affiliation{Departamento de F\'{\i}sica, Centro de Investigaci\'on y de Estudios Avanzados del Intituto Politécnico Nacional, Av. Intituto Politécnico Nacional 2508, San Pedro Zacatenco, M\'exico 07360, CDMX.
}%

\date{\today}

\begin{abstract}
In this work, we generalize Ernst potentials to the Einstein-Maxwell-Dilaton case and explicitly write the corresponding potential space metric. Since this metric is five-dimensional in potential space, we generalize the corresponding Newman-Penrose coefficients for this metric and compare this formalism with previous approaches to show that this formulation is very convenient for analyzing these spacetimes and finding new exact solutions. We show how to obtain old exact solutions and some new ones with very interesting properties.
\end{abstract}

\maketitle

\section{Introduction}
Frederick J. Ernst first introduced what are now known as the Ernst potentials in his seminal papers \cite{Ernst:1967wx, Ernst:1967by}, providing a powerful framework for reformulating the problem of axially symmetric gravitational fields within the context of General Relativity. Ernst’s approach employs a pair of complex scalar potentials to simplify the search for exact solutions by recasting the Einstein field equations as a system of coupled complex partial differential equations that are technically more tractable.

The original Ernst formalism has since been substantially generalized to encompass a broader class of spacetimes, including electrically charged configurations and those with a non-vanishing cosmological constant. In particular, Astorino \cite{Astorino:2012zm} extended Ernst’s method to incorporate a cosmological constant, thereby enabling the systematic “charging” of axisymmetric spacetimes and the construction of more general solutions relevant to black holes with non-asymptotically flat geometries. Building on this line of research, Charmousis et al.\ \cite{Charmousis:2006fx} derived explicit solutions describing rotating spacetimes in the presence of a cosmological constant, thereby demonstrating how the Ernst potentials can be generalized and modified to investigate both realistic astrophysical configurations and rigorously defined theoretical models.

In this work, we derive the Ernst equations in the presence of a scalar field of dilatonic or phantom type coupled to the electromagnetic field, and within an arbitrary gravitational theory of interest, which will be specified in the following paragraphs. We generalize the equations originally obtained by Matos et al. and establish the consistency between the two formulations, noting that the reformulation proposed by Matos encompasses the original Ernst formalism.

We start by introducing the Lagrangian in the \textit{Einstein frame}, which will be employed throughout this work and follows the formulation used in previous studies such as \cite{Matos:2000za,Matos:2010pcd}.
\begin{equation}\label{Lagrangiano}
    \mathfrak{L}=\sqrt{-g}\bigg(-\mathcal R +2\epsilon_0 (\nabla \phi)^2 + e^{-2 \alpha_0 \phi } F^2 \bigg),
\end{equation}
here, \(g = \det(g_{\mu\nu})\) denotes the determinant of the metric tensor, \(\mathcal R\) is the Ricci scalar, and \(F^{2} = F^{\mu\nu} F_{\mu\nu}\) represents the standard Maxwell invariant. The field \(\phi\) is a scalar field, while \(\epsilon = \pm 1\) distinguishes between the dilaton scalar field and the phantom scalar field, respectively. The coupling constant \(\alpha_{0} \in \{0, 1, 3, \tfrac{1}{2\sqrt{3}}\}\) characterizes the specific theoretical framework under consideration; for instance, these values correspond, respectively, to the Einstein–Maxwell (EM) theory, low-energy effective superstring (SS) theory, Kaluza–Klein (KK) theory, and entanglement relativity (ER) \footnote{See \cite{Minazzoli:2025gyw,Minazzoli:2025nbi} .}.
The corresponding field equations are:
\begin{subequations}\label{EcuacionesDeCampoOriginales}
\begin{equation}\label{Eq:Campo1}
    \nabla_\mu \left( e^{-2\alpha_0 \phi} F^{\mu \nu} \right)=0,
\end{equation}
\begin{equation}\label{Eq:Campo2}
    \epsilon_0 \nabla^2 \phi+\frac{\alpha_0}{2}  \left( e^{-2\alpha_0 \phi} F^{2} \right)=0,
\end{equation}
{\small
\begin{equation}\label{Eq:Campo3}
    \mathcal R_{\mu \nu}=2 \epsilon_0 \nabla_\mu \phi \nabla_\nu \phi  + 2  e^{-2\alpha_0 \phi} \left( F_{\mu \sigma} \tensor{F}{_\nu}{^\sigma} -\frac{1}{4} g_{\mu \nu } F^2 \right),
\end{equation}
}
\end{subequations}

In this setting, we focus on a stationary and axisymmetric space-time. As a result, the geometry admits two Killing vectors, $Killing = \{\partial_t, \partial_\varphi\}$. Taking these symmetries into account, we can make use of the following Weyl anzat metric:
\begin{equation}\label{ds Cilindricas}
     ds^2=-f\left ( dt-\omega d \varphi \right )^2+f^{-1} \left ( e^{2k} (d\rho ^2 + dz^2) +\rho^2 d\varphi^2 \right ), 
\end{equation}
In this framework, the metric functions $\{f,\omega,\kappa\}$ depend on the coordinates $(\rho,z)$, and we adopt the axisymmetric four-potential ansatz $ A_{\mu} = \big[ A_t(\rho,z),\,0,\,0,\, A_{\varphi \, L}(\rho,z) \big]$.

The formulation in spheroidal coordinates is provided in Appendix \ref{Apendix:Coordinates protection}.

\section{Potentials method}
For the purposes of this work, it is convenient to introduce two mutually orthogonal operators (i.e., satisfying $D\Tilde{D}=0$) in order to simplify the subsequent calculations, namely:
\begin{equation}\label{OperadoresDDguiño}
    D\equiv \begin{bmatrix}
        \partial_\rho \\
        \partial_z
    \end{bmatrix}, \qquad 
    \Tilde{D}=
    \begin{bmatrix}
    \begin{array}{c}
    \partial_z \\
    -\partial_\rho
    \end{array}
    \end{bmatrix},
\end{equation}
which has the following properties
\begin{align*}
   &\tilde{\tilde{D}} =-D, \qquad  \Tilde{D}A\Tilde{D}B=DADB, \\
   &\Tilde{D}ADB=-\Tilde{D}BDA, \quad \Tilde{D}\Tilde{D}=DD=\partial_\rho ^2 +\partial_z^2,
\end{align*}
and we shall adopt the notation $DDf=D^2f$, and $Df Df=Df^2$.
Considering the operators (\ref{OperadoresDDguiño}), the definition $\kappa^2 := e^{-2\alpha_0 \phi}$, and the metric (\ref{ds Cilindricas}), and performing the rescaling $d\varphi \rightarrow d(L\, \varphi)$, which implies $A_{\varphi \, L} = L\, A_\varphi$, we can reformulate (\ref{EcuacionesDeCampoOriginales}) as follows:
{\small
\begin{subequations}\label{EcuacionesDeCampoV1}
\begin{center}
    \text{Klein-Gordon Equation:}
    \begin{multline}
        D^2\kappa +D\kappa \left( \frac{D\rho}{\rho} - \frac{D\kappa}{\kappa} \right) \\
        -\frac{f \kappa^3 \alpha_0^2 }{\rho^2 \epsilon_0 }  \left[ L^2(\frac{\omega}{L} DA_t +DA_\varphi)^2-\frac{\rho^2}{f^2} DA_t^2 \right] =0 \label{eq:KleinGordonV1} ,
    \end{multline}
    \text{Maxwell Equations:}
    \begin{align}
        D \left[ \frac{4 \kappa^2 f}{\rho} L (\frac{\omega}{L} DA_t +DA_\varphi) \right]&=0, \label{Eq:Maxwell1V1} \\ 
        D \left[ 4 \kappa^2 \left( \frac{f \omega}{\rho} L (\frac{\omega}{L} DA_t +DA_\varphi) -\frac{\rho}{f} DA_t \right) \right]&=0, \label{Eq:Maxwell2V1}
    \end{align}  
    \text{Einstein Equations:}
    \begin{multline}
        D^2f + Df\left[ \frac{D\rho}{\rho} - \frac{Df}{f} \right]+\frac{f^3}{\rho^2}D\omega^2 \\
        -\frac{2  \kappa^2 f^2}{\rho^2}\left[ L^2(\frac{\omega}{L} DA_t +DA_\varphi)^2 +\frac{\rho^2}{f^2} DA_t^2 \right]=0, \label{Eq:Einstein1V1}
    \end{multline}
    \begin{equation}
        D^2\omega - D\omega \left[ \frac{D\rho}{\rho} - \frac{2 Df}{f} \right]+\frac{4  \kappa^2}{f}  L (\frac{\omega}{L} DA_t +DA_\varphi) DA_t=0. \label{Eq:Einstein2V1}
    \end{equation}
\end{center}
\end{subequations}
}

It is important to emphasize that the constant associated with the scalar field sector and its underlying theory is only relevant within the framework of the Klein–Gordon equation. Several remarks are in order. First, defining $[EM]^{\nu} := \nabla_\mu \left( e^{-2\alpha_0 \phi} F^{\mu \nu} \right)$, equation (\ref{Eq:Maxwell1V1}) is derived from the $\varphi$-component $[EM]^{\varphi} = 0$, whereas equation (\ref{Eq:Maxwell2V1}) follows from the time component $[EM]^{t} = 0$. Second, defining $[EE]_{\mu \nu} := R_{\mu \nu} - 2 \epsilon_0 \nabla_\mu \phi \nabla_\nu \phi - 2 e^{-2\alpha_0 \phi} \left( F_{\mu \sigma} \tensor{F}{_\nu}{^\sigma} - \frac{1}{4} g_{\mu \nu} F^2 \right)$, equation (\ref{Eq:Einstein1V1}) is obtained from the $tt$-component $[EE]_{tt} = 0$, while equation (\ref{Eq:Einstein2V1}) is obtained from the linear combination $[EE]_{\varphi t} + (\omega/L )\,[EE]_{tt} = 0$.

The field equations (\ref{EcuacionesDeCampoV1}) are derived from the Euler–Lagrange equation
\begin{equation}\label{Eq:EulerLagrange}
    D\left(  \frac{\partial \mathfrak{L}}{\partial(DY^A)} \right)
    - \frac{\partial \mathfrak{L}}{\partial Y^A} = 0,
\end{equation}
applied to the Lagrangian density
\begin{multline}\label{LagrangianoTesis2}
    \mathfrak{L}
    = \frac{\rho}{2f^2} Df^2
      - \frac{f^2}{2\rho} D\omega^2
      + \frac{2 \rho \epsilon_0 }{\alpha_0^2 \kappa^2}D\kappa^2 \\
      + \frac{2f \kappa^2}{\rho} 
        \left[
            L^2\left(\frac{\omega}{L } DA_t + DA_\varphi\right)^2
            - \frac{\rho^2}{f^2} DA_t^2
        \right],
\end{multline}
with the set of dynamical variables given by $Y^{A} = \{f,\omega,A_t,A_\varphi,\kappa\}$.
\subsection{Definition of the potentials}
Employing equations (\ref{Eq:Maxwell1V1}) and (\ref{Eq:Einstein2V1}), one can introduce two correspondingly significant potentials
\begin{subequations}\label{DefinicionPotencialesGeneral}
\begin{align}
    &\Tilde{D}\chi=\frac{2f\kappa^2}{\rho}  L(\frac{\omega}{L} DA_t +DA_\varphi), \label{TDchi}\\
    &\Tilde{D} \epsilon = \frac{f^2}{\rho} D\omega + \psi \Tilde{D}\chi \label{TDepsilon},
\end{align}
\end{subequations}
i.e. the existence of these two potentials is equivalent to fulfilling the first Maxwell equation and the second Einstein equation, and where $\psi=2A_t$. 
Consequently, the corresponding \textbf{potentials} can be denoted by $Y^A$.
\begin{equation}\label{Potenciales}
    (Y^A)^T=\Big[f,\epsilon,\psi,\chi,\kappa \Big],
\end{equation}
which are respectively the gravitational, rotational, electric, magnetic and scalar potentials.
Substituting the definitions (\ref{DefinicionPotencialesGeneral}) into the field equations (\ref{EcuacionesDeCampoV1}), these reduce to
{\small
\begin{subequations}\label{EcuacionesDeCampoV2}
\begin{center}
    \text{Klein-Gordon equation}
    \begin{equation}
        D(\rho D\kappa ) - \frac{\rho }{\kappa} D\kappa^2+\frac{\rho \kappa^3 \alpha_0^2}{4 f \epsilon_0}  \left( D\psi^2-\frac{1}{\kappa^4 } D\chi^2 \right) =0 \label{eq:KleinGordonV2} ,
    \end{equation}
    \text{Maxwell equations}
    \begin{equation}\label{Eq:Maxwell1V2}
        D (\rho D\psi)+\rho D\psi \left( \frac{2D\kappa}{\kappa} -\frac{Df}{f} \right)  -\frac{\rho}{f\kappa^2 } (D\epsilon -\psi D\chi) D\chi =0 ,
    \end{equation}
    \begin{equation}\label{Eq:Maxwell2V2}
        D (\rho D\chi)-\rho D\chi \left( \frac{2D\kappa}{\kappa} +\frac{Df}{f} \right)  +\frac{\rho \kappa^2}{f} (D\epsilon -\psi D\chi) D\psi =0, 
    \end{equation}  
    \text{Einstein equations}
    \begin{multline}\label{Eq:Einstein1V2}
        D (\rho Df)+\frac{\rho}{f}\left( (D\epsilon-\psi D\chi )^2 -Df^2 \right)  \\
        -\frac{\rho \kappa^2}{2} \left( D\psi^2 +\frac{1}{\kappa^4 } D\chi^2 \right) = 0 ,
    \end{multline}
    \begin{equation}\label{Eq:Einstein2V2}
        D (\rho (D\epsilon-\psi D\chi)) - \frac{2 \rho }{f} (D\epsilon-\psi D\chi) Df =0. 
    \end{equation}  
\end{center}
\end{subequations}
}
Equations (\ref{Eq:Maxwell2V2}) and (\ref{Eq:Einstein2V2}) emerge from the property of analyticity of the metric functions $A_t$ and $\omega$, i.e. $\Tilde{D} D A_t=0$, and $\Tilde{D} D \omega=0$.

The previously introduced field equations in (\ref{EcuacionesDeCampoV2}) can be rederived from the Euler–Lagrange expression in (\ref{Eq:EulerLagrange}) by employing the following Lagrangian density:
\begin{multline}\label{LagrangianoTesis3}
    \mathfrak{L}=\frac{\rho}{2f^2}\left( Df^2 + (D\epsilon-\psi D\chi)^2\right) + \frac{2 \rho \epsilon_0 }{\alpha_0^2 \kappa^2}D\kappa^2 \\
    -\frac{\rho \kappa^2}{2f} \left( D\psi^2 +\frac{1}{\kappa^4 } D\chi^2 \right).
\end{multline}
From (\ref{LagrangianoTesis3}), one obtains the line element of the corresponding potential space:
\begin{multline}\label{ds Transformado}
    ds^2=\frac{1}{2f^2}\left( df^2 + (d\epsilon-\psi d\chi)^2\right) + \frac{2 \epsilon_0 }{\alpha_0^2 \kappa^2}d\kappa^2 \\
    -\frac{\kappa^2}{2f}  \left( d\psi^2 +\frac{1}{\kappa^4 } d\chi^2 \right),
\end{multline}
where the metric is expressed in terms of the potentials \(f\), \(\epsilon\), \(\psi\), \(\chi\), and \(\kappa\).
The space generated by the potentials (\ref{Potenciales}) is characterized by a constant scalar curvature given by the Ricci invariant
\begin{equation}\label{EscalarRicci}
    \mathcal R=-12-\frac{\alpha_0^2}{\epsilon_0}.
\end{equation}
Consequently, one can conclude that this potential space is maximally symmetric and conformally flat.
\subsection{Matos Anzat}
An alternative ansatz was proposed in \cite{Matos:2000za}. Consider the complex combination variable obtained from the differentials of the potentials $Y^A$. From the combination $DY^A$, we define:
\begin{subequations}\label{Variables ABC}
    \begin{equation}\label{VariableA}
        A=\frac{1}{2f}\left[ Df -i (D\epsilon - \psi D\chi) \right],
    \end{equation}
    \begin{equation}\label{VariableB}
        B=-\frac{1}{2\sqrt{f}}\left[  \kappa D\psi -i \frac{1}{  \kappa} D\chi \right],
    \end{equation}
    \begin{equation}\label{VariableC}
        C=-\frac{D\kappa}{\kappa}.
    \end{equation}
\end{subequations}
From equation (\ref{VariableB}) it follows that the quantity $-2\sqrt{f}\, B$ can be interpreted as an \textit{electromagnetic coupling with a scalar-field–dependent pseudo-force} involving $\kappa$. This interpretation is motivated by the classical description in which the force 1-form is given by $F = dV$, where $V$ denotes the potential and $d$ is the exterior derivative operator.

More precisely, $DY^A$ should not be interpreted as a physical force, rather, it is an abstract construct introduced to provide a conceptual interpretation of this quantity. The variable $C$ may be regarded as a pseudo-force generated by either a scalar dilatonic field or a phantom field, while the term $2f\,A$ is analogous to an \textit{Ernst pseudo-force}.
By employing the definitions given in (\ref{Variables ABC}), field equations (\ref{EcuacionesDeCampoV2}) can be recast in the form:
\begin{subequations}\label{EcuacionesDeCampo-ABC}
 \begin{center}
     \text{Field equations in terms of $\{ A,B,C\}$:}
 \end{center}
    \begin{equation}
        \frac{1}{\rho} D(\rho A)=B\overline{B} +A(A-\overline{A}),\label{EcuacionDeCampoA}
    \end{equation}
    \begin{equation}
        \frac{1}{\rho} D(\rho B)=-\frac{B}{2}(A-3\overline{A})+C\overline{B},\label{EcuacionDeCampoB}
    \end{equation}
    \begin{equation}
        \frac{1}{\rho} D(\rho C)=\frac{\alpha_0^2}{2\epsilon_0}(B^2+\overline{B}^2),\label{EcuacionDeCampoC}
    \end{equation}
\end{subequations}
\section{Relations with Ernst formulation}
This section aims to clarify the correspondence between the Ernst formulation and the extended reformulation \cite{Matos:1994hm,Matos:2000za,Matos:2010pcd}. By carefully examining the works \cite{Astorino:2012zm,Ernst:1967by,Ernst:1967wx}, the notation and variables can be consistently identified as follows:
\begin{align*}
    \omega &\quad \rightarrow \quad -\omega ,\\
    (\nabla F\, \,; \, \nabla^2F) &\quad \rightarrow \quad \Big( D\, F\, ;\, \frac{1}{\rho}D(\rho D F)  \Big),\\
    \nabla F \cdot \nabla G &\quad \rightarrow \quad DFDG ,\\
    \alpha e^{\Omega/2} &\quad \rightarrow \quad f ,\\
    \alpha &\quad \rightarrow \quad \rho ,\\
    e^{2\nu}/\sqrt{\alpha} &\quad \rightarrow \quad e^{2k}/f ,\\
    \hat{e}_{\varphi}\times\Vec{\nabla} F &\quad \rightarrow \quad \Tilde{D} F,\\
    A_0 &\quad \rightarrow \quad \psi/2 ,\\
\end{align*}
where \(F\) and \(G\) denote arbitrary variables.
Hence, by matching the corresponding differential equations and variables, we infer that this extended reformulation is related to the Ernst formulation in the following way:

\begin{subequations}\label{RelacionesErnstMatos}
\begin{center}
    \text{Ernst formulation to extended reformulation}
    \end{center}
    \begin{equation}\label{RelacionErnstChi}
        \Tilde{A}_3 \quad \rightarrow \quad \chi/2,
    \end{equation}
    \begin{equation}\label{RelacionErnstPotEMComplejo}
        2\Phi \equiv \psi + i\chi,
    \end{equation}
    \begin{equation}\label{RelacionErnstTwist}
        h \equiv \frac{1}{2} \psi \chi -\epsilon,
    \end{equation}
    {\small
    \begin{equation}\label{RelacionErnstPotencial}
        \varepsilon \quad \equiv f - \Phi \overline{\Phi} +ih= (f-i\epsilon)- \frac{1}{4}(\psi^2+\chi^2)+ \frac{i}{2} \psi \chi ,
    \end{equation}}
\end{subequations}

where $\varepsilon$ is the \textbf{Ernst potential}, $h$ is the Ernst twist potential, $\Phi$ is the Ernst electromagnetic complex potential.
\subsection{Interpretations of the variables A,B}
Assuming a constant scalar field with null coupling, $\kappa = 1$, vanishing constant parameter $\alpha_0=0$, and natural units, the relations in \eqref{Variables ABC} reduce to
{\small
\begin{equation*}
    2fA = D(f - i\epsilon) + i\psi\, D\chi, 
    \quad 
    -\sqrt{f}\, B = \frac{1}{2}D(\psi - i\chi), 
    \quad 
    C = 0.
\end{equation*}
}
Expressed in terms of the Ernst potentials, these relations become
\begin{align*}
    2fA 
    &= D\bigl(\varepsilon + \Phi \overline{\Phi}\bigr) 
      - \frac{i}{2}D(\psi \chi) 
      + i\psi\, D\chi \notag \\
    &= D\varepsilon + \Phi\, D\overline{\Phi} + \overline{\Phi}\, D\Phi 
       + \frac{i}{2}\bigl(\psi\, D\chi - \chi\, D\psi\bigr) \notag \\
    &= D\varepsilon + 2\,\overline{\Phi}\, D\Phi,
\end{align*}
where $ \overline{\Phi}\, D\Phi - \Phi\, D\overline{\Phi}=\frac{i}{2}\bigl(\psi\, D\chi - \chi\, D\psi\bigr)$, and
\begin{equation*}
    -\sqrt{f}\, B = D\overline{\Phi}.
\end{equation*}
Consequently, the explicit relations between the auxiliary variables $A, B$ and the Ernst potentials are given by
\begin{subequations}
    \begin{equation}\label{fErnst}
        f \equiv \mathfrak{Re}(\varepsilon) + \Phi\, \overline{\Phi},
    \end{equation}
    \begin{equation}\label{AenTerminosdeErnst}
        A \equiv \frac{1}{2f} \bigl[ D\varepsilon + 2\,\overline{\Phi}\, D\Phi \bigr],
    \end{equation}
    \begin{equation}\label{BenTerminosdeErnstsinkapa}
        B\big|_{\kappa=1} \equiv -\frac{1}{\sqrt{f}}\, D\overline{\Phi}.
    \end{equation}
\end{subequations}
\subsection{Generalized Ernst equations within extended reformulation}
To establish the equivalence between these two formulations, we derive the Ernst equations for the pair $(\varepsilon, \Phi)$ by starting from the generalized Ernst equations in the presence of an electromagnetic field coupled to a scalar field.
\subsubsection{Generalized Ernst equation for $\Phi$}
Using the definition \eqref{TDchi}, we obtain $\Tilde{D}\left(\frac{\chi}{2}\right)= \frac{\kappa^2 f}{\rho}\bigl(\omega D A_t + D A_\varphi\bigr)$.
Substituting this relation into \eqref{Eq:Maxwell2V1} yields
\begin{equation}\label{EcuacionParaErnstPhi1}
    D \bigg[ \omega \Tilde{D} \chi -\frac{ \rho \kappa^2}{f} D\psi \bigg]=0.
\end{equation}
Next, invoking the property $\Tilde{\Tilde{D}}=-D$ in the definition \eqref{TDchi}, we can solve for $\Tilde{D}A_\varphi$ and obtain
\begin{equation}\label{DAvarphi tilde}
    2 \Tilde{D} A_\varphi = -\bigg[ \frac{ \rho }{ \kappa^2 f} D\chi +\omega \Tilde{D} \psi \bigg].
\end{equation}
Then, applying the identity $D\Tilde{D}=0$ to \eqref{DAvarphi tilde}, we arrive at
\begin{equation}\label{EcuacionParaErnstPhi2}
    D \bigg[ \frac{ \rho }{ \kappa^2 f} D\chi +\omega \Tilde{D} \psi \bigg]=0.
\end{equation}
By multiplying equation \eqref{EcuacionParaErnstPhi2} by $i$ and subsequently adding the result to equation \eqref{EcuacionParaErnstPhi1}, we obtain
\begin{equation}\label{EcuacionParaErnstPhi3}
    D\bigg[ \omega \tilde{D}\Phi -\frac{i}{f} \rho \, B_L \bigg]=0,
\end{equation}
where we introduce the new variable
\begin{equation}\label{BLeonel}
    2B_L = \kappa^2  D\psi +\frac{i}{\kappa^2 }D \chi.
\end{equation}
If $\kappa=1 $ implies $ \Rightarrow \quad  B_L= D\Phi$

By performing straightforward manipulations, equation \eqref{EcuacionParaErnstPhi3} can be recast in the form
\begin{align*}
    &\frac{i f}{\rho} \bigg(
    D\omega\, \Tilde{D}\Phi
    - \frac{i}{f} D(\rho B_L)
    + \frac{\rho}{f^2} Df\, B_L
    \bigg) = 0 \\
    &\text{implies} \quad \frac{1}{\rho} D(\rho B_L)
    = \frac{D f}{f}\, B_L
    + i \frac{f}{\rho} \,\Tilde{D}\omega\, D\Phi,
\end{align*}
where $-\,\Tilde{D}\omega\, D\Phi=D\omega\, \Tilde{D}\Phi$.
Moreover, from \eqref{VariableA}, \eqref{TDepsilon}, and \eqref{RelacionesErnstMatos}, it follows that the following relations hold:
\begin{subequations}
    \begin{equation}\label{RelacionDf A}
        D f = f \,(A + \overline{A}),
    \end{equation}
    \begin{equation}
        \Tilde{D}\omega
        = \frac{\rho}{f^2} \left(\psi\, D\chi - D\epsilon \right),
    \end{equation}
    \begin{align}\label{RelacionDEpsilon psiDchi A}
        i\,\frac{f^2}{\rho}\,\Tilde{D}\omega
        &= -i\,(D\epsilon - \psi\, D\chi)
        = f\,(A - \overline{A}) \notag \\
        &= D\big(i h - \Phi \overline{\Phi}\big) + 2\,\overline{\Phi}\, D\Phi.
    \end{align}
\end{subequations}

Substituting these identities into the previous expression, we finally obtain \textbf{the generalized Ernst equation for the complex electromagnetic potential} $\Phi$:
\begin{equation}\label{ErnstGeneralizadaPHI}
    \frac{1}{\rho} D(\rho B_L)
= (B_L + D\Phi)\,A + (B_L - D\Phi)\,\overline{A},
\end{equation}
this equation expresses one of the main objectives of this paper.
By setting $\kappa=1$, which means excluding the scalar field, we achieve $D\Phi=B_L$. Substituting this value into \eqref{ErnstGeneralizadaPHI} yields:
\begin{align}\label{EcuacionErnstPhi}
    \frac{1}{\rho} D(\rho D\Phi)=2A D\Phi=\frac{D\varepsilon +2\overline{\Phi } D\Phi}{\mathfrak{Re}(\varepsilon) + \Phi \overline{\Phi}} D\Phi ,
\end{align}
this is the known Ernst equation for the variable $\Phi$.
\subsubsection{Generalized Ernst equations for $\varepsilon$}
To obtain the generalized Ernst equation for the variable $\varepsilon$, we will add the term $A(A+\overline{A})=A Df /f=\rho A Df /(\rho f)$ to both sides of \eqref{EcuacionDeCampoA}, then one can observe that:
\begin{align*}
    &A(A+\overline{A})+(A-\overline{A})+B \overline{B} =B\overline{B}+2A^2\\
    &= \frac{\rho A Df}{\rho f}+\frac{1}{\rho }D(\rho A)=\frac{1}{\rho f } D (\rho A D f).
\end{align*}
Consequently, the second objective of this study has been achieved and consists in deriving the \textbf{generalized Ernst equation for the Ernst potential} $\varepsilon$, which is given by:
\begin{equation}\label{ErnstGeneralizadaEpsilon}
    \frac{1}{\rho} D(\rho A f)= f \Big( B \overline{B} +2A^2 \Big),
\end{equation}
and it is precisely the same equation \eqref{EcuacionDeCampoA}, only written in a different representation.
By setting $\kappa=1$, and considering \eqref{EcuacionErnstPhi}, it is evident $fB \overline{B}=D\Phi D\overline{\Phi}$, thus
{\small
\begin{align}
    &\frac{1}{\rho} D(\rho A f)=\frac{1}{2\rho } D(\rho \varepsilon)+\frac{1}{\rho} D(\rho \overline{\Phi} D\Phi ) \notag \\
    &\qquad \qquad \quad=A(D \varepsilon +2 \overline{\Phi} D\Phi) +f \overline{B} B \notag \\
    &\therefore \quad \frac{1}{\rho}D(\rho \varepsilon)-2A D\varepsilon = 2\Big[ 2A  \overline{\Phi} D\Phi -\frac{1}{\rho} D(\overline{\Phi}  \, \, \rho D\Phi)  + f B \overline{B }\Big] \notag \\
    &\qquad \qquad \qquad \qquad \qquad \qquad =0 \notag \\
    &\text{implies} \qquad \frac{1}{\rho}D(\rho \varepsilon)= 2A D\varepsilon=\frac{D\varepsilon +2\overline{\Phi } D\Phi}{\mathfrak{Re}(\varepsilon) + \Phi \overline{\Phi}} D\varepsilon .
\end{align}
}
which is the well-known Ernst equation for the variable $\varepsilon$, and where $2A  \overline{\Phi} D\Phi -\frac{1}{\rho} D(\overline{\Phi}  \, \, \rho D\Phi)=-D\Phi D \overline{\Phi}$.
\section{Covariant form the field equations}
The equations (\ref{EcuacionesDeCampoV2}) can be expressed in a covariant form
\begin{equation}\label{EcuacionDeCampoCovariante}
    D(\rho DY^A)+\rho \tensor{\hat{\Gamma}}{^A}{_{BC}} DY^B DY^C=0,
\end{equation}
where $\tensor{\hat{\Gamma}}{^A}{_{BC}}$ represent the Christoffel symbols associated with the potential space constituted by $Y^A$. Equation \eqref{EcuacionDeCampoCovariante} corresponds to a geodesic equation defined on the potential space.

The adoption of the covariant formulation of the field equations preserves their invariant structure and permits the introduction of a general coordinate system $(\lambda,\tau)$. By incorporating the functional dependencies $\rho(\lambda,\tau)$ and $z(\lambda,\tau)$ into equation (\ref{EcuacionDeCampoCovariante}), one obtains
\begin{subequations}
\begin{equation}\label{Eq:GeodesicaYA}
    \tensor{Y}{^A}{_{,i:j}}+\tensor{\hat{\Gamma}}{^A}{_{BC}} \tensor{Y}{^B}{_{,i}} \tensor{Y}{^C}{_{,j}}=0,
\end{equation}
\begin{equation}\label{Eq:LaplaceGeneraLambda}
    D(\rho D\lambda^l)+\rho \tensor{\Gamma}{^l}{_{ij}} D\lambda^i D\lambda^j =0.
\end{equation}
\end{subequations}
where $\lambda^i=\{\lambda,\tau\}$, $Y_{,l}:=\partial Y/\partial \lambda^l$, and $\tensor{\Gamma}{^l}{_{ij}}$ denote the Christoffel symbols associated with the general subspace \textit{generated} by $\lambda$ and $\tau$. 
Equation (\ref{Eq:LaplaceGeneraLambda}) represents the \emph{Laplace equation in curved subspaces} and specifies the condition that the general subspace must satisfy.

By introducing the complex variable $\xi \equiv \lambda+i \tau $, one is able to perform summation $D\xi =D\lambda+iD\tau $ and $D \, \overline{\xi} =D\lambda-i D\tau $, thereby transforming the equation (\ref{Eq:LaplaceGeneraLambda}) into a new form:
\begin{subequations}\label{Eq:LaplaceXi}
\begin{align}
    \frac{1}{\rho}D(\rho D \xi) &= \frac{-2\sigma}{1-\sigma \xi \, \overline{\xi}} \, \overline{\xi} \, (D \xi )^2 , \\
    \frac{1}{\rho}D(\rho D \, \overline{\xi}) &= \frac{-2\sigma}{1-\sigma \xi \, \overline{\xi}} \, \xi \, (D \overline{\xi} )^2 ,
\end{align}
\end{subequations}
where
\begin{align*}
     \tensor{\Gamma}{^\lambda}{_{\lambda \lambda}} =-\tensor{\Gamma}{^\lambda}{_{\tau \tau}}=\tensor{\Gamma}{^\tau}{_{\lambda \tau}}=\frac{2\sigma}{1-\sigma \xi \, \overline{\xi}} \, \lambda, \\
     \tensor{\Gamma}{^\tau}{_{\tau \tau}} =-\tensor{\Gamma}{^\tau}{_{\lambda \lambda}}=\tensor{\Gamma}{^\lambda}{_{\lambda \tau}}=\frac{2\sigma}{1-\sigma \xi \, \overline{\xi}} \, \tau.
\end{align*}
\subsection{Carateristics of the general subspace}
Recalling that the potential space is maximally symmetric due to its constant curvature, it follows that the \textit{general subspace} generated by $\xi$ and $\overline{\xi}$ must also be maximally symmetric, implying that its curvature is constant. This condition leads to the following most general form of the metric:
\begin{equation}\label{MetricaEspacioGeneral}
    ds^2=\frac{2(d\lambda^2+d\tau^2)}{(1-\sigma [\lambda^2+\tau^2] )^2}=\frac{d\xi d \overline{\xi}}{(1-\sigma \xi \overline{\xi})}, 
\end{equation}
where the Ricci invariant $\mathcal R \propto \sigma$, that is, $\sigma$ characterizes the curvature of this general subspace.

It is important to emphasize that this \textbf{two-dimensional general subspace} $(\xi, \overline{\xi})$ does not coincide with the \textbf{physical space $(\rho,z)$}. For example, although $\sigma=0$ implies that the general subspace $(\xi, \overline{\xi})$ is flat, the corresponding physical subspace $(\rho,z)$ need not be flat.
\section{Extended field equations}
In this section, we accomplish the third and main objective of this paper, namely the generalization of the results of \cite{Matos:2000za,Matos:2010pcd} to the case of a dilaton/phantom scalar field coupled to the electromagnetic field, within an arbitrary underlying theory, while allowing for a non-flat geometry in the $\{\lambda,\tau\}$-space.
\subsection{Second Anzat}
Using the complex variables $\xi$ and $\overline{\xi}$, we introduce the following second ansatz, motivated by the earlier work of Matos:
\begin{subequations}\label{AnzatABC}
\begin{align}
    A&=\mathtt{x}_1(\xi, \overline{\xi}\,) D\xi+\mathtt{x}_2(\xi, \overline{\xi}\,) D\, \overline{\xi},\\
    B&=\mathtt{y}_1(\xi, \overline{\xi}\,) D\xi+\mathtt{y}_2(\xi, \overline{\xi}\,) D\, \overline{\xi},\\
    C&=\mathtt{z}_1(\xi, \overline{\xi}\,) D\xi+\mathtt{z}_2(\xi, \overline{\xi}\,) D\, \overline{\xi}.
\end{align}
\end{subequations}
Accordingly, expression \eqref{AnzatABC} can be written as $U = u_1 D\xi + u_2 D\,\overline{\xi}$, which leads to
\begin{align*}
    \frac{1}{\rho}D(\rho U)
    &=\frac{1}{\rho} \Big[ u_1 D(\rho D\xi)+ u_2 D(\rho D\, \overline{\xi} \,)\Big]\\
    &+\Big[ D\xi Du_1+ D \,\overline{\xi}\, Du_2 \Big] \\
    &= \Bigg[ u_1 \bigg(\ln \Big\{ (1-\sigma\,\xi\,\overline{\xi})^{2} u_1 \Big\}\bigg)_{,\xi} \Bigg] D\xi^2 \\
    &+ \Bigg[ u_2 \bigg(\ln \Big\{ (1-\sigma\,\xi\,\overline{\xi})^{2} u_2 \Big\}\bigg)_{,\overline{\xi}} \Bigg] D \overline{\xi}\, ^2 \\
    & \qquad \qquad +(u_{1,\overline{\xi}}+u_{2,\xi}) D\xi D\, \overline{\xi},
\end{align*}
so that we may substitute $\rho^{-1}D(\rho U)$ for each $U \in \{ A,B,C \}$, together with the ansatz \eqref{AnzatABC} into the field equations \eqref{EcuacionesDeCampo-ABC}. By comparing the coefficients of the independent quadratic forms $D\xi^2$, $D\, \overline{\xi} \,^2$, and $D\xi D\, \overline{\xi}$, we obtain
\begin{widetext}
\begin{subequations}\label{N_EcuacionesDeCampo-abc}
    \begin{center}
    \text{Expressions of the extended field equations in terms of $\{ \mathtt{x},\mathtt{y},\mathtt{z}\}$}
    \end{center}
    \begin{align}
        \mathtt{x}_1(\mathtt{x}_1-\overline{\mathtt{x}}_2)+\mathtt{y}_1\overline{\mathtt{y}}_2&=\mathtt{x}_1\bigg(\ln \Big\{ (1-\sigma\,\xi\,\overline{\xi})^{2} \mathtt{x}_1 \Big\}\bigg)_{,\xi}, \label{N_EC-a1} \\
        \mathtt{x}_2(\mathtt{x}_2-\overline{\mathtt{x}}_1)+\mathtt{y}_2\overline{\mathtt{y}}_1&=\mathtt{x}_2\bigg(\ln \Big\{ (1-\sigma\,\xi\,\overline{\xi})^{2} \mathtt{x}_2 \Big\}\bigg)_{,\overline{\xi}}, \label{N_EC-a2} \\
        \mathtt{x}_1(\mathtt{x}_2-\overline{\mathtt{x}}_1)+\mathtt{x}_2(\mathtt{x}_1-\overline{\mathtt{x}}_2)+(\mathtt{y}_1\overline{\mathtt{y}}_1+\mathtt{y}_2\overline{\mathtt{y}}_2)&=(\mathtt{x}_1)_{,\overline{\xi}}+(\mathtt{x}_2)_{,\xi}, \label{N_EC-a3}
    \end{align}
    \begin{align}
        \frac{\mathtt{y}_1}{2}(3\overline{\mathtt{x}}_2-\mathtt{x}_1)+\mathtt{z}_1\overline{\mathtt{y}}_2&=\mathtt{y}_1\bigg(\ln \Big\{ (1-\sigma\,\xi\,\overline{\xi})^{2} \mathtt{y}_1 \Big\}\bigg)_{,\xi}, \label{N_EC-b1} \\
        \frac{\mathtt{y}_2}{2}(3\overline{\mathtt{x}}_1-\mathtt{x}_2)+\mathtt{z}_2\overline{\mathtt{y}}_1&=\mathtt{y}_2\bigg(\ln \Big\{ (1-\sigma\,\xi\,\overline{\xi})^{2} \mathtt{y}_2 \Big\}\bigg)_{,\overline{\xi}}, \label{N_EC-b2} \\
        \frac{\mathtt{y}_1}{2}(3\overline{\mathtt{x}}_1-\mathtt{x}_2)+\frac{\mathtt{y}_2}{2}(3\overline{\mathtt{x}}_2-\mathtt{x}_1)+\mathtt{z}_1\overline{\mathtt{y}}_1+\mathtt{z}_2\overline{\mathtt{y}}_2&=(\mathtt{y}_1)_{,\overline{\xi}}+(\mathtt{y}_2)_{,\xi}, \label{N_EC-b3}
    \end{align}
    \begin{align}
        \frac{\alpha_0^2}{2\epsilon_0}(\tensor{\mathtt{y}}{_1}{^2}+\tensor{\overline{\mathtt{y}}}{_2}{^2})=\mathtt{z}_1\bigg(\ln \Big\{ (1-\sigma\,\xi\,\overline{\xi})^{2} \mathtt{z}_1 \Big\}\bigg)_{,\xi},\label{N_EC-c1} \\
        \frac{\alpha_0^2}{2\epsilon_0}(\tensor{\mathtt{y}}{_2}{^2}+\tensor{\overline{\mathtt{y}}}{_1}{^2})=\mathtt{z}_2\bigg(\ln \Big\{ (1-\sigma\,\xi\,\overline{\xi})^{2} \mathtt{z}_2 \Big\}\bigg)_{,\overline{\xi}},\label{N_EC-c2} \\
        \frac{\alpha_0^2}{\epsilon_0}(\mathtt{y}_1 \mathtt{y}_2+\overline{\mathtt{y}}_1\overline{\mathtt{y}}_2)=(\mathtt{z}_1)_{,\overline{\xi}}+(\mathtt{z}_2)_{,\xi},\label{N_EC-c3} 
    \end{align}
\end{subequations}
\end{widetext}
The comparison is valid because the terms $D\xi^2$, $D\,\overline{\xi}^{\,2}$, and $D\xi\,D\,\overline{\xi}$ are linearly independent. This independence follows from the fact that they are constructed from the independent basis vectors $\partial_\rho, \partial_z$ and the independent variables $(\xi,\overline{\xi})$. As a result, these terms constitute a new basis that spans a three-dimensional subspace. When one sets $\sigma = 0$, thereby restricting attention to an abstract flat subspace, the resulting system reduces to the equations presented presented in \cite{Matos:2000za,Matos:2010pcd}.

The relations between the potentials and the variables $\mathtt{x,y,z}$ transform to:
\begin{subequations}\label{N_RelacionesPotenciales-abc}
    \begin{center}
    \text{Relations between the potentials and $\{ \mathtt{x,y,z}\}$}
    \end{center}
    \text{\textbf{$\mathtt{x}$:}}
    \begin{align}
        \mathtt{x}_1=\frac{1}{2f}\big[ f_{,\xi} -i (\epsilon_{,\xi} - \psi \chi_{,\xi}) \big],\label{N_a1} \\
        \mathtt{x}_2=\frac{1}{2f}\big[ f_{,\overline{\xi}} -i (\epsilon_{,\overline{\xi}} - \psi \chi_{,\overline{\xi}}) \big],\label{N_a2} 
    \end{align}
    \text{\textbf{$\mathtt{y}$:}}
    \begin{align}
        \mathtt{y}_1=-\frac{1}{2\sqrt{f}}\big[  \, \kappa \, \psi_{,\xi} -\frac{i}{\,\kappa} \chi_{,\xi} \big],\label{N_b1} \\
        \mathtt{y}_2=-\frac{1}{2\sqrt{f}}\big[  \, \kappa \,  \psi_{,\overline{\xi}} -\frac{i}{  \,\kappa} \chi_{,\overline{\xi}} \big],\label{N_b2} 
    \end{align}
    \text{\textbf{$\mathtt{z}$:}}
    \begin{align}
        \mathtt{z}_1=-\frac{1}{\kappa} \kappa_{,\xi},\label{N_c1} \\
        \mathtt{z}_2=-\frac{1}{\kappa} \kappa_{,\overline{\xi}}\label{N_c2}.
    \end{align}
\end{subequations}
\paragraph{Real scalar field existence.} To guarantee the existence and smoothness of the scalar field defined by $\alpha_0\,\phi=-\ln{\kappa}$, we must impose the condition $\partial_\xi \partial_{\,\overline{\xi}}\,\phi=\partial_{\,\overline{\xi}}\partial_\xi \phi$, which is equivalent to $\partial_{\,\overline{\xi}}\,z_1=\partial_\xi\,z_2$. Moreover, if we choose $\phi$ to be a \emph{real} scalar field, we must require $\phi=\overline{\phi}$, which in turn implies $z_1=\overline{z_2}$.
\begin{theorem}
    In the non-flat $(\xi,\overline{\xi})$-space $\sigma \neq 0$ and using the anzat \eqref{AnzatABC} and within Einstein–Maxwell theory (with $\alpha_0 = 0$), no solutions exist that correspond to configurations with a real scalar field $\phi$ coupled to the electromagnetic field, nor do such solutions arise in any theory with $\sigma \neq 0$ in which the electromagnetic field is absent (i.e., $\mathtt{y} = 0$).
\end{theorem}
\begin{proof}
    Let 
    \[
        S(\xi,\bar\xi)=(1-\sigma\,\xi\bar\xi)^2,
        \qquad \sigma\in\mathbb{R},\ \sigma\neq 0.
    \]
    and we know that $\kappa^2 = e^{-2\alpha_0 \, \phi}$, then $\mathtt{z}_1=\alpha_0 \, \phi_{,\xi}$ and $\mathtt{z}_2=\alpha_0 \, \phi_{,\overline{\xi}}$. 
    
    Suppose there exists a \emph{real} scalar potential $\phi(\xi,\bar\xi)\in\mathbb{R}$ of class $C^2$ such that
    \begin{equation}\label{Potencial real}
    z_1=\alpha_0\,\partial_\xi\phi,
    \qquad
    z_2=\alpha_0\,\partial_{\bar\xi}\phi,
    \qquad
    \alpha_0\in\mathbb{R}\ \text{constant}.
    \end{equation}
    
    Since $\phi$ is real, \eqref{Potencial real} implies the reality condition
    \begin{equation}\label{eq:reality}
    z_1=\overline{z_2}.
    \end{equation}
    Moreover, the very existence of $\phi$ implies the integrability identity
    \begin{equation}\label{eq:integrabilidad}
    \partial_{\bar\xi}z_1=\partial_{\bar\xi}\partial_\xi(\alpha_0\phi)
    =\partial_\xi\partial_{\bar\xi}(\alpha_0\phi)
    =\partial_\xi z_2,
    \end{equation}
    because mixed partial derivatives commute for $C^2$ functions.

    Under $y_1=y_2=0$ or $\alpha_0=0$, the block $Z$ of \eqref{N_EcuacionesDeCampo-abc} reduces to the system
    \begin{subequations}
    \begin{align}
    0&=\mathtt{z}_1\big(\ln(S \mathtt{z}_1)\big)_{,\xi}, \label{eq:Z1 0}\\
    0&=\mathtt{z}_2\big(\ln(S \mathtt{z}_2)\big)_{,\bar\xi}, \label{eq:Z2 0}\\
    0&=(\mathtt{z}_1)_{,\bar\xi}+(\mathtt{z}_2)_{,\xi}. \label{eq:Z3 0}
    \end{align}
    \end{subequations}

    Adding and subtracting \eqref{eq:Z3 0} with \eqref{eq:integrabilidad} we obtain
    \begin{equation}\label{eq:Z dependiente}
    \partial_{\bar\xi} \mathtt{z}_1=0,
    \qquad
    \partial_\xi \mathtt{z}_2=0.
    \end{equation}
    Thus $\mathtt{z}_1$ is holomorphic and $\mathtt{z}_2$ is anti-holomorphic:
    \[
    \mathtt{z}_1=\mathtt{z}_1(\xi),
    \qquad
    \mathtt{z}_2=\mathtt{z}_2(\bar\xi).
    \]

    Now take \eqref{eq:Z1 0}. On any open set where $\mathtt{z}_1\neq 0$, we may divide by $\mathtt{z}_1$ and conclude
    \begin{equation}\label{eq:log}
    \big(\ln(S \mathtt{z}_1)\big)_{,\xi}=0
    \quad\Longrightarrow\quad
    (\ln \mathtt{z}_1)_{,\xi}=-(\ln S)_{,\xi}.
    \end{equation}
    Since $S=(1-\sigma\xi\bar\xi)^2$, we compute
    \[
    (\ln S)_{,\xi}
    =
    2\left(\ln(1-\sigma\xi\bar\xi)\right)_{,\xi}
    =
    -\frac{2\sigma\bar\xi}{1-\sigma\xi\bar\xi}.
    \]
    Therefore \eqref{eq:log} becomes
    \begin{equation}\label{eq:contr}
    (\ln \mathtt{z}_1)_{,\xi}
    =
    \frac{2\sigma\bar\xi}{1-\sigma\xi\bar\xi}.
    \end{equation}

    \emph{Here is the contradiction:}
    By \eqref{eq:Z dependiente} we know that $z_1=z_1(\xi)$, so $(\ln \mathtt{z}_1)_{,\xi}$ depends solely on $\xi$. In contrast, the right-hand side of \eqref{eq:contr} depends explicitly on $\bar\xi$ whenever $\sigma\neq 0$. An identity of this form cannot hold on any open set. Thus, there can be no point at which $\mathtt{z}_1\neq 0$. By continuity of $\mathtt{z}_1$, we must therefore have
    \[
    \mathtt{z}_1\equiv 0.
    \]
    Finally, from \eqref{eq:reality} we obtain $\mathtt{z}_2=\overline{\mathtt{z}_1}\equiv 0$.

    This completes the proof of the theorem.

\end{proof}

\section{Known solution}\label{Known solution}

\subsection{Kerr solution}
In this section, we solve equations \eqref{N_EcuacionesDeCampo-abc} to derive the Kerr solution in a simple and transparent manner. We start by imposing $\mathtt{y}_1=\mathtt{y}_2=\mathtt{z}_1=\mathtt{z}_2$=0, meaning that we look for a configuration without a scalar field, without an electromagnetic field, and with a chiral ansatz $\mathtt{x}_2=0$ such that only $\mathtt{x}_1 \neq 0$ remains. Under these assumptions, equations \eqref{N_EcuacionesDeCampo-abc} reduce to the following form:
\begin{subequations}\label{Ec Dif de Kerr}
    \begin{align}
    &\bigg(\ln \Big[ \Delta^{2} \mathtt{x} \Big] \bigg)_{,\xi}=\mathtt{x} \label{PrimeraEcuacionKerr}, \\
        &\bigg(\ln\mathtt{x}\bigg)_{,\overline{\xi}}+\overline{\mathtt{x}}=0 \label{SegundaEcuacionKerr},
    \end{align}
\end{subequations}
where
\begin{equation}\label{eq:Funcion-Inv-Real_Delta}
    \Delta=(1-\sigma\,\xi\bar\xi),
        \qquad \sigma\in\mathbb{R},\ \sigma\neq 0.
\end{equation}
and 
\[
    \Delta_{,\xi}=-\sigma\, \overline{\xi}, \qquad \Delta_{,\xi}=-\sigma\, \xi . 
\]
Observe that $\partial_{\xi}=\partial_\lambda + i\partial_{\lambda}$ and $\partial_{\,\overline{\xi}}=\partial_\lambda - i\partial_{\lambda}$. Hence, $\overline{\partial_{\, \xi}}$ coincides with $\partial_{\, \overline{\xi}}$, which implies $\overline{\partial a}=\partial_{\,\overline{\xi}}\,\overline{a}$. Therefore, by taking the complex conjugate of equation (\ref{SegundaEcuacionKerr}), we obtain:
\begin{subequations}
\begin{align}
    &\bigg(\ln \Big[ \Delta^{2} \mathtt{x} \Big]\bigg)_{,\xi}=\mathtt{x}=-\bigg(\ln\overline{\mathtt{x}} \bigg)_{,\xi}, \notag \\
    & \Rightarrow \qquad \bigg(\ln \Big[ \Delta^{2} \mathtt{x} \, \overline{ \mathtt{x}} \, \Big]\bigg)_{,\xi}=0 \label{Condicion1Kerr}, \\
    & \therefore \qquad h(\overline{\xi})=\mathtt{x} \, \overline{\mathtt{x}} \, \Delta^2. \label{funcion h de kerr}
\end{align}
\end{subequations}
By inserting the expression from (\ref{funcion h de kerr}) into the conjugate of (\ref{Condicion1Kerr}), we clearly obtain that $h(\overline{\xi})=1:$ constant. Thus, we obtain the condition
\begin{align} \label{funcion h de kerr 1} 
    & 1=\mathtt{x} \, \overline{\mathtt{x}} \, \Delta^2, 
\end{align}
By substituting equation (\ref{funcion h de kerr 1}) into equation (\ref{SegundaEcuacionKerr}), we obtain $(\mathtt{x})_{,\overline{\xi}}=-1/\Delta^2$, thus
\begin{equation} \label{funcio g Kerr}
    \mathtt{x}=\frac{-1/(\sigma \, \xi)}{\Delta} +g(\xi).
\end{equation}
Now, by substituting equation (\ref{funcio g Kerr}) into (\ref{PrimeraEcuacionKerr}), we obtain a new differential equation for the function $g(\xi)$.
\begin{align}
    &\xi ^2 \sigma ^2 \Delta^2 (g'-g^2) +\sigma -1 \notag \\
    &+2 g \xi  \sigma \Delta \left( 1-\sigma^2 \, \xi \, \overline{\xi} \right)=0, \notag \\
    & \therefore \quad g(\xi)= \frac{1-\xi  \, \overline{\xi} \,  \sigma ^2}{\xi  \sigma \Delta}-\frac{1}{\xi +l_1} \label{funcion g1 Kerr},
\end{align}
where $l_1$ is an integration constant that correspond to the mass parameter of the compact object.
Now, using (\ref{funcion g1 Kerr}), we obtain the solution of (\ref{Ec Dif de Kerr}) corresponding to Kerr with $\sigma \neq 0$, this solution corresponds to the third class of solutions given in \cite{Matos:2000ai}.
\begin{equation}\label{xyzdeKerr}
    \mathtt{x}_{K}=\frac{(l_1 \, \sigma \, \overline{\xi}+1)/(\xi+l_1)}{(\sigma \, \xi \,\overline{\xi}-1)}.
\end{equation}
Finally, the constants relating the parameters are obtained from equation \eqref{SegundaEcuacionKerr}, which yields the relation
\begin{equation}
    l_1=\pm 1/\sqrt{\sigma}.
\end{equation}
Therefore, by setting $\mathtt{x}_2=\mathtt{y}_1=\mathtt{y}_2=\mathtt{z_1}=\mathtt{z}_2=0$, we find that the potentials $\psi, \chi, \kappa$ are constant, and the remaining equations \eqref{N_RelacionesPotenciales-abc} reduce to the form
    \begin{align*}
        \Big( \ln f \Big)_{, \,\xi} &= \mathtt{x}, \qquad \Big( \ln f \Big)_{, \, \overline{\xi}} =\overline{\mathtt{x}}, \\
         \epsilon_{, \,\xi} &=i\mathtt{x}f, \qquad \epsilon_{, \, \overline{\xi}} =-i \, \overline{\mathtt{x}} \, f,
\end{align*}
whose solutions are
\begin{subequations}\label{Potenciales de Kerr}
    \begin{align}
        &\psi_K= \chi_K=\kappa_K=1, \\
        &f_K=f_0 \, \frac{\sigma \,\xi \, \overline{\xi}-1}{(\xi+l_1)(\overline{\xi}+l_1)}, \label{fKerrPotencial}\\
        &\epsilon_K=i\, \frac{f_0}{l_1}\,\frac{\xi-\overline{\xi}}{(\xi+l_1)(\overline{\xi}+l_1)} \label{epsilonKerrPotencial}.
    \end{align}
\end{subequations}
If we employ the $\delta = 1$-Tomimatsu–Sato solution, given by $\xi = L x - i a y$, which satisfies equation \eqref{Eq:LaplaceXi} in prolate \footnote{Subject to the condition $|\mathcal{M}_\infty|^{2}>a^2+Q_{L}^{2}+H_{L}^{2}$, i.e. sub-extreme case, see Appendix \ref{Apendix:Coordinates protection}.} spheroidal coordinates ( $L_{-}\, x = r - l_1$ and $y = \cos\theta$), where $(r,\theta)$ denote the Boyer–Lindquist coordinates, then it is necessary to choose $f_0 = l_1^{\, 2}$ in order to recover the metric function $f$ corresponding to the Kerr solution in these coordinates. In summary, using $\xi =\xi_K= L x - i a y$ and $f_0= 1/\sigma=l_1^{\,\, 2}$, $f_K$ will be
\begin{equation*}
    f_{K}=\frac{r^2-2l_1 \,r +a^2 \cos^2{\theta}}{r^2+a^2 \cos^2{\theta}}.
\end{equation*}

\subsection{Kerr Newman Solution}
In this case, we again employ the chiral ansatz in the gravitational-twist sector with $\mathtt{x}_2=0$. However, we now impose this chiral ansatz in the electromagnetic sector by setting $\mathtt{y}_1=0$, and we take the scalar field to be constant, $\mathtt{z}_1=\mathtt{z}_2=0$, thereby decoupling the scalar field and working within the Einstein–Maxwell theory with $\alpha_0^2=0$. Under these assumptions, substituting into the equations \eqref{N_EcuacionesDeCampo-abc}
\begin{subequations}\label{EcuacionesDeKerrNewman}
    \begin{equation}\label{EcuacionXKerr1}
         \mathtt{x}_{,\xi}-\frac{2\sigma\overline{\xi}}{\Delta}\,\mathtt{x}=\mathtt{x}^2,
    \end{equation}
    \begin{equation}\label{EcuacionXKerr2}
         \mathtt{x}_{,\overline{\xi}}=\mathtt{y}\overline{\mathtt{y}}-\mathtt{x}\overline{\mathtt{x}},
    \end{equation}
    \begin{equation}\label{EcuacionYKerr1}
         \mathtt{y}_{,\overline{\xi}}-\frac{2\sigma\xi}{\Delta}\,\mathtt{y}=\frac{3}{2}\overline{\mathtt{x}}\,\mathtt{y}.
    \end{equation}
    \begin{equation}\label{EcuacionYKerr2}
         \mathtt{y}_{,\xi}=-\frac{\mathtt{x}}{2}\,\mathtt{y},
    \end{equation}
\end{subequations}
where $\Delta$ is defined in \eqref{eq:Funcion-Inv-Real_Delta}.

Equation \eqref{EcuacionXKerr1} is a standard Riccati equation, and its solution is given by 
\begin{equation}\label{eq:x_general_form}
\mathtt{x}(\xi,\overline{\xi})
=
\frac{-\sigma\overline{\xi}}{\Delta\big(1-\sigma\overline{\xi}\,F(\overline{\xi})\,\Delta\big)}.
\end{equation}

\paragraph{Canonical choice for obtaining a rational 1-soliton-type solution.}
A particularly convenient choice (in the sense that it yields a simple pole at $\xi = -l_1$) is
\begin{equation}\label{eq:F_choice}
F(\overline{\xi})
=
\frac{1}{\sigma\overline{\xi}\big(1+\sigma l_1\overline{\xi}\big)},
\qquad l_1\in\mathbb{R}.
\end{equation}
By substituting expression \eqref{eq:F_choice} into the general form given in \eqref{eq:x_general_form}, and subsequently carrying out a straightforward algebraic simplification, we obtain the explicit form of the $\mathtt{x}$ variable corresponding to the Kerr–Newman solution:
\begin{equation}\label{eq:x_KerrNewman}
\mathtt{x}_{KN}
=-
\frac{1+\sigma l_1\overline{\xi}}{(\xi+l_1)\,\Delta}.
\end{equation}
Using \eqref{eq:x_KerrNewman} into \eqref{EcuacionYKerr2}, we obtain the solution 
\begin{equation}\label{eq:y_intermediate}
\mathtt{y}(\xi,\overline{\xi})
=
g(\overline{\xi})\,
\sqrt{\frac{\sigma \, (\xi+l_1)}{(\sigma\,\xi\bar\xi-1)}},
\end{equation}
where $g(\overline{\xi})$ is an integration function that depends on $\overline{\xi}$. 

By substituting expressions \eqref{eq:y_intermediate} and \eqref{eq:x_KerrNewman} into the differential equation \eqref{EcuacionYKerr1}, we derive the following differential equation $\frac{3}{l_1+\overline{\xi}}+\frac{2g}{g}=0$, for the function $g$, whose solution is
\begin{equation}\label{eq:g_solution}
g(\overline{\xi})
=
\frac{Q_0}{(\overline{\xi}+l_1)^{3/2}},
\qquad
Q_0\in\mathbb{C},
\end{equation}
where we define $Q_0 = q + i p$ with $q,p \in \mathbb{R}$, interpreting $q$ as the electric charge and $p$ as the magnetic charge, both expressed in units of length.

We can thus ultimately derive the $\mathtt{y}$ variable corresponding to the \textbf{Kerr Newman} solution
\begin{equation}\label{eq:y_KerrNewman}
\mathtt{y}_{KN}
=
\frac{Q_0\,\sqrt{\sigma(\xi+l_1)}}{(\overline{\xi}+l_1)^{3/2}\sqrt{(\sigma\,\xi\bar\xi-1)}}.
\end{equation}
Taking \eqref{EcuacionXKerr2} into account, and upon substituting \eqref{eq:x_KerrNewman} and \eqref{eq:y_KerrNewman}, we obtain the following constraint on the parameter:
\begin{equation}\label{eq:ConstriccionParametros Kerr Newman}
\frac{1}{\sigma}=l_1^{\, \, 2}-|Q_0|^2.
\end{equation}
It illustrates how the mass, together with the electric and magnetic charges, distorts the curvature of the $(\xi,\overline{\xi})$-space.

By substituting $\mathtt{z}_1=\mathtt{z}_2=\mathtt{y}_1=\mathtt{x}_2=0$ into equations~\eqref{N_RelacionesPotenciales-abc}, we obtain the following relations, which determine the corresponding potentials of the Kerr–Newman solution
\begin{align*}
    \Big( \ln f \Big)_{, \,\xi} &= \mathtt{x}, \quad \Big( \ln f \Big)_{, \, \overline{\xi}} =\overline{\mathtt{x}}, \\
    -i\epsilon_{, \,\xi} &=\mathtt{x}f + \psi \sqrt{f} \,\overline{\mathtt{y}}, \\
    i\epsilon_{, \, \overline{\xi}} &=\overline{\mathtt{x}}f + \psi \sqrt{f} \,\mathtt{y}, \\
    -\sqrt{f}\, \mathtt{y}&=\psi_{,\overline{\xi}}=-i \chi_{,\overline{\xi}}, \\
    -\sqrt{f}\, \overline{\mathtt{y}}&=\psi_{,\xi}=i \chi_{,\xi}. \\
\end{align*}
whose solutions are
\begin{subequations}
\begin{align}
    &\kappa=1, \\
    f_{KN}&=\frac{f_0 \left(\sigma \, \xi \bar{\xi }-1\right)}{\left(\xi +l_1 \right) \left(\bar{\xi }+l_1\right)}, \\
    \psi_{KN}&=\psi_0+\sqrt{f_0 \, \sigma} \left( \frac{\overline{Q}_0}{\xi +l_1} +\frac{ Q_0}{\overline{\xi} +l_1}\right)  , \\
    \chi_{KN}&=\chi_0+i\sqrt{f_0 \, \sigma} \left( \frac{ Q_0}{\overline{\xi} +l_1}- \frac{\overline{Q}_0}{\xi +l_1}\right)  , \\
    \epsilon_{KN}&=\epsilon_1+\frac{i f_0 \, \sigma}{2}  \bigg( \frac{Q_0^{\,\, 2} -2l_1 (\xi+l_1)}{(\xi+l_1)^2} \notag \\
    &\qquad \qquad- \frac{\overline{Q}_0^{\,\, 2} -2l_1 (\overline{\xi}+l_1)}{(\overline{\xi}+l_1)^2} \bigg).
\end{align}
\end{subequations}
The relation \eqref{eq:ConstriccionParametros Kerr Newman} guarantees the existence of the $\epsilon$-potential. If this parameter constraint is not satisfied, $\epsilon$ is not integrable. Thus, to integrate ($\epsilon_{,\xi},\epsilon_{,\overline{\xi}}$), we must impose \eqref{eq:ConstriccionParametros Kerr Newman}.

Analogously to the Kerr case, to obtain \(f\) for the Kerr–Newman solution by applying \(\xi = Lx - iay\) in prolate spheroidal coordinates, we impose \(f_{0} = 1/\sigma\) and use \(\xi = \xi_{K}\), which yields
\begin{equation*}
    f_{KN}= \frac{r^{2}-2l_1\, r+|Q_0|^{2}+a^{2}\cos^{2}\theta}{r^{2}+a^{2}\cos^{2}\theta}.
\end{equation*}
\subsection{Kerr-Dilaton in Kaluza-Klein theory}
In 1995, in the paper \cite{Rasheed:1995zv}, Rasheed derives the rotating Kaluza–Klein black hole with a dilaton $\epsilon_0=+1$ of fixed coupling $\alpha_0=\sqrt3$ by taking the Kerr solution as a seed and exploiting the sigma-model symmetries that arise when stationary, axisymmetric 5D gravity is reduced to 4D. More precisely, he performs a restricted $SO(1,2)$ transformation consisting of two boosts characterized by parameters $\alpha,\beta$ and a rotation chosen so as not to induce NUT charge which turns the Kerr solution into one carrying both electric and magnetic charges. He then reconstructs the full 5D metric and, via the Kaluza–Klein reduction ansatz, extracts the corresponding 4D metric, Maxwell field, and dilaton. This procedure yields a closed form, explicit family of rotating dyonic black hole solutions in Kaluza–Klein theory.

In our notation, and within the sub-extreme regime, refer to Appendix \ref{Apendix:Coordinates protection} \footnote{See \cite{Bixano:2025bio,Bixano:2025qxp})}, the definition of spheroidal prolate ($-$) coordinates is given by $L_{-}\, x=r-l_1$, $y=\cos{\theta}$, and their relation to Weyl coordinates is specified by $\rho=L_{-}\sqrt{(x^2-1)(1-y^2)}, \, z=L_{-}\, xy$. 

Therefore, using \( \xi=Lx+iay,\qquad \bar\xi=Lx-iay,
\qquad
Lx=\frac{\xi+\bar\xi}{2},\qquad
y=\frac{\xi-\bar\xi}{2ia}, \) and the following definitions
{\small
\begin{equation}\label{eq:Fun-Inv_Z_Sigma}
    \begin{aligned}
        & r=l_1+Lx=l_1+\frac{\xi+\bar\xi}{2},\notag\\
        & \Sigma=r^2+a^2y^2=(\xi+l_1)(\bar\xi+l_1),\notag\\
        &\mathbf Z=\frac{2l_1 r}{\Sigma}=l_1\left(\frac{1}{\xi+l_1}+\frac{1}{\bar\xi+l_1}\right),\notag\\
        &\mathbf B=\sqrt{1+\frac{v^2}{1-v^2}\,Z}
        =\sqrt{1+\frac{v^2}{1-v^2}\,l_1\left(\frac{1}{\xi+l_1}+\frac{1}{\bar\xi+l_1}\right)}.\notag\\[2ex]
    \end{aligned}
\end{equation}
}
where $|v|<1$, we can reconstruct the potentials associated with the solution \cite{Rasheed:1995zv} using \eqref{FuncionesMetricas-xy}, obtaining:
{\small
\begin{equation}
    \begin{aligned}
        &f_{KD}=\frac{1-\mathbf Z}{\mathbf B}, \qquad \kappa_{KD}=\mathbf B^{3/2}, \\
        & \psi_{KD}=\frac{v}{1-v^2}\,\frac{\mathbf Z }{\mathbf B ^2}, \\
        & \chi_{KD}=\chi_0+\frac{l_1 v_b s}{i}\,\frac{\xi-\bar\xi}{\Sigma}, \\
        &\epsilon_{KD}=\epsilon_\infty+\frac{l_1 s}{i}\,\frac{\xi-\bar\xi}{\Sigma}.
    \end{aligned}
\end{equation}
}
\section{Other solution}

In 1994, in the article Ref.~\cite{Matos:1994hm}, Matos derived two solutions for specific values of $\alpha_0^2$, assuming a dilatonic-type scalar field with $\epsilon_0=+1$. The solution associated with \eqref{Potenciales de Kerr} is equivalent to matrix (20) in \cite{Matos:1994hm}, under the parametrization
\[
\lambda=\sqrt{\sigma}\,\xi, \qquad \tau=\sqrt{\sigma}\,\overline{\xi}, \qquad l_1=1/\sqrt{\sigma},
\]
with $\alpha_0^2=0$. Meanwhile, matrix (21) in \cite{Matos:1994hm} yields a different solution, we have used \eqref{eq:Funcion-Inv-Real_Delta}, \eqref{eq:Fun-Inv_Z_Sigma} and introducing two additional functions
\begin{equation}\label{eq:Funcion-Inv-PI}
    \begin{aligned}
        &\Pi=(\xi-l_1)(\bar\xi-l_1), \qquad \delta=\xi-\bar\xi,
    \end{aligned}
\end{equation}
to obtain its corresponding potentials:
{\small
\begin{equation}\label{eq:Nueva-Sol_Matos_21}
    \begin{aligned}
        &f^2 =\frac{\Delta^6}{d_\sigma\,\Sigma^4\,R_\sigma}, \qquad \kappa^{2/3}=-\frac{2}{\Delta}\sqrt{\frac{R_\sigma}{d_\sigma}}, \\
        &\epsilon = -\frac{b_\sigma}{d_\sigma}\,\frac{\delta^2}{\Sigma^2}, \qquad \chi =-2\sqrt2\,\frac{e_\sigma}{d_\sigma}\,\frac{\delta}{\Sigma}, \\
        &\psi =\frac{\sqrt2\,\delta}{4R_\sigma} \left(d_\sigma c_\sigma\,\Pi-\frac{b_\sigma e_\sigma}{\Sigma}\,\delta^2\right),
    \end{aligned}
\end{equation}
}
where we have set \( R_\sigma = d_\sigma f_0(\Delta^2 - 2\sigma\delta^2) - e_\sigma^2\delta^2 \), and
{\small
\[
    a_\sigma=\frac{a_0}{l_1^4},\quad
    b_\sigma=\frac{b_0}{l_1^2},\quad
    c_\sigma=\frac{c_0}{l_1^3},\quad
    d_\sigma=\frac{d_0}{l_1^4},\quad
    e_\sigma=\frac{e_0}{l_1^3}.
\]
}
The constants $a_0,b_0,c_0,d_0,e_0,f_0$ are constants whose values were adopted from \cite{Matos:1994hm}.

Starting from the Matos $SO(2,1)$ subclass of the chiral model $SL(3,\mathbb{R})/SO(2,1)$ presented in \cite{Matos:1994hm}, and making the identification $\lambda,\tau \quad \mapsto \quad \xi,\overline{\xi}$, we obtained the potentials \eqref{eq:Nueva-Sol_Matos_21}. If we consider the Kerr-Tomimatsu-Sato solution $\xi=Lx+iay$, we have
{\small
\begin{equation}
    \begin{aligned}
    &\Pi=(r-2l_1)^2+a^2y^2=(\xi-l_1)(\bar\xi-l_1),\\
    &\Delta=1-\sigma\,\xi\bar\xi=-\sigma\Big((L_{\pm}x)^2+a^2y^2-l_1{}^2 \Big)\\ 
    & \qquad \qquad =-\sigma \Big( r^2-2l_1\, r +a^2 y^2\Big),\\
    &\delta=2i ay=\xi-\bar\xi, \\
    &R_\sigma(x,y)=
    d_\sigma f_0\bigg[
    \sigma^2\Big(  r^2-2l_1\, r +a^2 y^2 \Big)^2 \\
    &\qquad \qquad  +8\sigma a^2y^2 \bigg] +4e_\sigma^{\,2}a^2y^2.
\end{aligned}
\end{equation}
}
Similarly, by imposing that the metric functions $f,\psi,\kappa$ be real and requiring the integrability of equation \eqref{EcuacionesDiferencialesOmega}, one possible choice is to enforce
{\small
\begin{equation}\label{eq:Condicion_Int_omega_Matos21}
    d_0=f_0=b_0\in\mathbb{R}, \quad e_0=\pm i \sqrt{2 d_0 f_0}, \quad c_0=2b_0 f_0 /e_0,
\end{equation}
}
thereby obtaining the dyonic branch $(A_t\neq 0,A_\varphi\neq 0)$ of the Bonnor dipole/black dihole family, originally identified in 1996 in \cite{Bonnor:1966dipole}, see also \cite{Lukacs:1977gw,Garcia-Compean:2015ywa,Emparan:1999au,Becerril:1990ek,Becerril:1990ek}. The corresponding metric functions with $f_0=1$ in our formulation are:
{\footnotesize
\begin{equation}
    \begin{aligned}
        f&=-\left(\frac{L^{2}x^{2}+a^{2}y^{2}-\,l_1^{2}}{\,l_1^{2}+2Ll_1x+L^{2}x^{2}+a^{2}y^{2}\,}\right)^{\!2},\\
        \omega&=\omega_0=0, \qquad \kappa=-2\sqrt{2},\\
        A_t&=\frac{\psi}{2}=\frac{a\,l_1\,y}{2\big(l_1^{2}+2Ll_1x+L^{2}x^{2}+a^{2}y^{2}\big)},\\
        A_\varphi &=-\frac{a\,l_1\,(l_1+Lx)\,(1-y^{2})}{2L^{2}\,\big(L^{2}x^{2}+a^{2}y^{2}-l_1^{2}\big)},\\
        e^{2k}&=\frac{\big(L^{2}x^{2}+a^{2}y^{2}-l_1^{2}\big)^{4}}{L^{8}\,(x^{2}-y^{2})^{4}}.
    \end{aligned}
\end{equation}
}

\section{Newman-Penrose formalism}\label{Newman-Penrose formalism} 
In this section we are going to analyze the geometric part of Ernst’s formulation and extended reformulation. Let the orthonormal coframe $\{\theta^A\}$  be
\begin{align*}
\theta^{0}&=\frac{d  f}{\sqrt{2}\,f}=\frac{1}{\sqrt{2}}d  (\ln{f}), 
&\theta^{1}=\frac{d \epsilon-\psi\,d \chi}{\sqrt{2}\,f},
\\
\theta^{2}&=\frac{\kappa}{\sqrt{2f}}\,d \psi, 
&\theta^{3}=\frac{1}{\kappa\sqrt{2f}}\,d \chi,
\\
\theta^{4}&=\frac{\sqrt2}{\alpha_0}\,\frac{d \kappa}{\kappa}=\frac{\sqrt{2}}{\alpha_0}d  (\ln{\kappa}),
\end{align*}
the metric \eqref{ds Transformado} becomes $d s_T^2=\eta_{AB}\,\theta^A\theta^B$, with $\eta_{AB}=\mathrm{diag}(+1,+1,-1,-1,\epsilon_0)$. However,when the five-dimensional analogue of the Newman–Penrose formalism is applied to the potential space, we obtain an interesting result
\begin{equation}\label{Base 1F Newman Penrose}
\begin{aligned}
\ell
&=\frac{1}{\sqrt2}\,(\theta^0+\theta^1)
=\frac{1}{2f}\Big(d  f+d \epsilon-\psi\,d \chi\Big),
\\
n
&=\frac{1}{\sqrt2}\,(\theta^0-\theta^1)
=\frac{1}{2f}\Big(d  f-d \epsilon+\psi\,d \chi\Big),
\\
m
&=\frac{1}{\sqrt2}\,(\theta^2+i \,\theta^3)
=\frac{1}{2\sqrt f}\Big(\kappa\,d \psi+i \,\kappa^{-1}d \chi\Big),
\\
\bar m
&=\frac{1}{\sqrt2}\,(\theta^2-i \,\theta^3)
=\frac{1}{2\sqrt f}\Big(\kappa\,d \psi-i \,\kappa^{-1}d \chi\Big),
\\
s
&=\theta^4
=\frac{\sqrt2}{\alpha_0}\,\frac{d \kappa}{\kappa}.
\end{aligned}
\end{equation}
Hence, the metric can be written as $ds^2 = 2\,\ell\,n - 2\,m\,\bar m + \epsilon_0\,s^2$. Consequently, the connection between \eqref{Base 1F Newman Penrose} and \eqref{Variables ABC} is expressed through the following relations
\begin{equation}
\begin{aligned}
A&=\frac12\Big[(1-i )\,\ell+(1+i )\,n\Big],
\\
\bar A&=\frac12\Big[(1+i )\,\ell+(1-i )\,n\Big], \\
    B&=-\,\bar m,
\qquad
\bar B=-\,m,
\\
C&=-\frac{\alpha_0}{\sqrt2}\,s.
\end{aligned}
\end{equation}
In other words, Ernst’s formulation is confined to the subspace spanned by $n,\ell,m,\bar m$ or, equivalently, by $A,B$. However, once a scalar field is introduced, the directions $m$ and $\bar m$ become deformed because of their interaction with the scalar field. Analogously, an additional fifth dimension $s$ appears, oriented along the direction defined by the \emph{scalar potential}, i.e. the extended reformulation lies in the complete 5 dimensions $n,\ell,m,\bar m,s$.

Similarly, while the one-form $A$ lies in the hyperplane spanned by $n,\ell$, that is, in the plane defined by the gravitational twist potentials, the one-form $B$ lies in the deformed electromagnetic direction $\bar m$.

Using the relation $d\theta^I + \omega^{I}{}_{J}\wedge\theta^J = 0,\quad \theta^I\in\{\ell, n, m, \bar m, s\}$, we can determine the associated connection one-forms in the absence of torsion given by
\begin{equation}\label{eq:1Formas conexion NewmanPenroseIndMixtos}
\begin{aligned}
\omega^{n}{}_{\ell}&=-n,
&
\omega^{\ell}{}_{n}&=-\ell,
\\
\omega^{n}{}_{m}&=i \,\bar m,
&
\omega^{\ell}{}_{\bar m}&=i \,m,
\\
\omega^{\bar m}{}_{\ell}&=-\frac12\,\bar m,
&
\omega^{\bar m}{}_{n}&=-\frac12\,\bar m,
\\
\omega^{m}{}_{\ell}&=-\frac12\,m,
&
\omega^{m}{}_{n}&=-\frac12\,m,
\\
\omega^{\bar m}{}_{s}&=\frac{\alpha_0}{\sqrt2}\,m,
&
\omega^{m}{}_{s}&=\frac{\alpha_0}{\sqrt2}\,\bar m,
\\
\omega^{s}{}_{m}&=\frac{\alpha_0}{\sqrt2\,\epsilon_0}\,m,
&
\omega^{s}{}_{\bar m}&=\frac{\alpha_0}{\sqrt2\,\epsilon_0}\,\bar m.
\end{aligned}
\end{equation}
the others that are absent here simply disappear. By employing the relation $\Omega^{I}{}_{J} = d\omega^{I}{}_{J} + \omega^{I}{}_{K} \wedge \omega^{K}{}_{J}$, we can determine the curvature 2-forms
\begin{equation}\label{eq:2Formas curvatura NewmanPenroseIndMixtos}
\begin{aligned}
\Omega^{n}{}_{\ell}
&=\ell\wedge n-\frac{i }{2}\,m\wedge\bar m,
\\
\Omega^{\ell}{}_{n}
&=-\ell\wedge n+\frac{i }{2}\,m\wedge\bar m,
\\
\Omega^{\bar m}{}_{\ell}
&=\frac14\,\ell\wedge\bar m-\frac14\,n\wedge\bar m+\frac{\alpha_0}{2\sqrt2}\,m\wedge s,
\\
\Omega^{\bar m}{}_{n}
&=-\frac14\,\ell\wedge\bar m+\frac14\,n\wedge\bar m+\frac{\alpha_0}{2\sqrt2}\,m\wedge s,
\\
\Omega^{m}{}_{\ell}
&=\frac14\,\ell\wedge m-\frac14\,n\wedge m+\frac{\alpha_0}{2\sqrt2}\,\bar m\wedge s,
\\
\Omega^{m}{}_{n}
&=-\frac14\,\ell\wedge m+\frac14\,n\wedge m+\frac{\alpha_0}{2\sqrt2}\,\bar m\wedge s,
\\
\Omega^{\bar m}{}_{s}
&=-\frac{\alpha_0}{2\sqrt2}(\ell+n)\wedge m-\frac{\alpha_0^2}{2}\,\bar m\wedge s,
\\
\Omega^{m}{}_{s}
&=-\frac{\alpha_0}{2\sqrt2}(\ell+n)\wedge \bar m-\frac{\alpha_0^2}{2}\,m\wedge s.
\end{aligned}
\end{equation}
the others that are absent here simply disappear. And the corresponding invariants are
{\small
\begin{equation}
    \begin{aligned}
    & \mathcal R=-12-\frac{\alpha_0^2}{\epsilon_0} \qquad \mathcal R_{AB} \mathcal R^{AB}=144+4\alpha_0^4 \\
    &\mathcal R_{ABCD}R^{ABCD}=192+12\alpha_0^4, \\
    &\mathcal C_{ABCD} \mathcal C^{ABCD}=96+\frac{16\alpha_0^2}{\epsilon_0}+\frac{22}{3}\,\alpha_0^4
\end{aligned}   
\end{equation}
}
where $\mathcal R$ is the Ricci scalar, $\mathcal R_{AB}$ is the Ricci tensor, $\mathcal R_{ABCD}$ is the Riemann tensor and $\mathcal C_{ABCD}$ is the Weyl tensor.
\section{Conclusions}   %
We extend Ernst’s formulation to include a scalar field coupled to the electromagnetic field in various theoretical frameworks and introduce the corresponding potential formulation originally used by Ernst. We also show the equivalence between the framework of Matos et al. \cite{Matos:1994hm,Matos:2000za,Matos:2010pcd}, defined on the space spanned by $n,\ell,m,\bar m,s$ or $A,B,C$, and Ernst’s formalism \cite{Astorino:2012zm,Ernst:1967by,Ernst:1967wx}, restricted to the subspace spanned by $n,\ell,m,\bar m$ or, equivalently, $A,B$.

A key outcome is a Newman–Penrose-type formalism on the potential target space, adapted to the five-dimensional potential metric. In this framework, the Einstein–Maxwell–Scalar-Field equations become geodesic equations on the potential space, enabling analysis of geodesic completeness and identification of geodesically complete subspaces. This NP-like viewpoint guides the construction of suitable ansätze for exact solutions.

We showed that the formalism systematically reproduces standard solutions such as Kerr and Kerr–Newman and naturally includes the rotating Kaluza–Klein dilatonic black-hole sector (Rasheed-type solutions) within the same potential-space framework. We also established a precise correspondence between the solution in \cite{Matos:1994hm} and the classical Bonnor solution in a specific parameter region, thereby clearly linking chiral/matrix constructions to earlier exact-solution families.

When $\alpha_0^2 = 3$ (KK with $\epsilon_0 = +1$), the target space is the symmetric coset $SL(3,\mathbb R)/SO(2,1)$. The Killing algebra is then the full $\mathfrak{sl}(3,\mathbb R)$ of dimension 8, rather than the 5-dimensional solvable subalgebra. This introduces genuinely non-gauge generators, allowing new solutions to be generated from a seed via Ehlers/Harrison transformations, which mix $(f,\epsilon,\psi,\chi,\kappa)$ nontrivially and thus also $(A,B,C)$.

For $\alpha_0^2 \neq 3$ (or more generally $\alpha_0^2 \neq 0$), the symmetry group is smaller and consists only of:

Gravitational dilatation:
\(
f\mapsto \lambda^2 f,\quad
\epsilon\mapsto \lambda^2 \epsilon,\quad
\psi\mapsto \lambda\,\psi,\quad
\chi\mapsto \lambda\,\chi,\quad
\kappa\mapsto \kappa, \quad 0<\lambda \in \mathbb{R}.
\)

Scalar field dilatation:
\(
\kappa\mapsto \mu\,\kappa,
\quad
\psi\mapsto \mu^{-1}\psi,
\quad
\chi\mapsto \mu\,\chi,
\quad
f,\epsilon\ \text{unchanged}, \quad 0<\mu \in \mathbb{R}.
\)

Electromagnetic duality:
\(
\psi\mapsto -\psi,\quad
\chi\mapsto -\chi,\quad
f,\epsilon,\kappa\ \text{unchanged}.
\)

And the three-dimensional Heisenberg algebra generated by
\(
K_\epsilon=\partial_\epsilon,
\quad
K_\chi=\partial_\chi,
\quad
K_\psi=\partial_\psi+\chi\,\partial_\epsilon,
\)
with
\(
[K_\psi,K_\chi]=K_\epsilon,
\quad
[K_\epsilon,\cdot]=0.
\)

\section{Acknowledgements}

LB thanks SECIHTI-M\'exico for the doctoral grant.
This work was also partially supported by SECIHTI M\'exico under grants SECIHTI CBF-2025-G-1720 and CBF-2025-G-176. 


\begin{appendices}
\section{Coordinates protection}\label{Apendix:Coordinates protection}
 The Weyl anzat metric \eqref{ds Cilindricas} expressed in spheroidal Oblates($+$)/Prolates($-$) coordinates $(x, y)$ is defined as
\begin{multline}\label{ds sp}
    ds^2 = -f\left( dt-\omega d \varphi \right)^2
     + \frac{(L_{\pm})^2}{f} \bigg( (x^2\pm1)(1-y^2) d\varphi^2 \\
     +(x^2\pm y^2) e^{2k} \left\{ \frac{dx^2}{x^2\pm 1} +\frac{dy^2}{1-y^2} \right\} \bigg).
\end{multline}
where we have substituted $L$ by $L_{\pm}$, and the Weyl coordinates are related to $(x, y)$ through the following relation
\begin{equation}\label{WeylSpheroidalCordiantes}
\rho=(L_{\pm})\sqrt{(x^2\pm1)(1-y^2)} , \qquad z=(L_{\pm})xy, 
\end{equation}
and $\rho \in [0,\infty)$, $\{ z ,x \}\in \mathbb{R}$, $y \in [-1,1]$, $(L_{\pm}) \geq 0$.

We are considering two scenarios

\paragraph{Sub-extreme (S-E: lower sign $-$) condition}:
\begin{align}\label{Sub-Extreme}
|\mathcal{M}_\infty|^{2}&>a^2+Q_{L}^{2}+H_{L}^{2} ,\\
|\mathcal{M}_\infty|^{2}&=L_{-}^2+a^2+Q_{L}^{2}+H_{L}^{2},
\end{align}

\paragraph{Super-extreme (SU-E: upper sign $+$) condition}:
\begin{align}\label{Super-Extreme}
|\mathcal{M}_\infty|^{2}&<a^2+Q_{L}^{2}+H_{L}^{2} ,\\
L_{+}^2+|\mathcal{M}_\infty|^{2}&=a^2+Q_{L}^{2}+H_{L}^{2}.
\end{align}
where \(a = J_\infty / |\mathcal{M}_\infty|\) denotes the angular momentum per unit effective mass, with \(\mathcal{M}_\infty = l_1 + i\,N_\infty\), \(l_1\) a length-scale parameter, and \(Q_L = Q_\infty\), \(H_L = H_\infty\) the electric and magnetic geometric charges, respectively, while \(N_\infty\) is the NUT parameter. The quantities \(Q_\infty, H_\infty, N_\infty, J_\infty\) are conserved invariant charges, defined in the standard classical sense as in \cite{Komar:1958wp,Nedkova:2011hx,Clement:2015aka,Clement:2022pjr,BallonBordo:2019vrn,Misner:1963fr,Manko:2005nm}. In general, the Komar mass \(M_\infty\) coincides with \(l_1\).
Finally, the Boyer-Lindquist coordinates $(r,\theta)$ are related to the previous coordinates as
\begin{equation}\label{BoyerLindquistSpheroidalCordiantes}
    (L_{\pm}) x=r-l_1, \qquad y=\cos{\theta},
\end{equation}
where $r \in (-\infty,-l_1]\cup[l_1,\infty)$, $\theta\in [0,\pi]$, and the variable $l_1=r_s/2$, where $r_s$ represents the Schwarzschild radius.
\subsection{Potential differential equations}
Using the relation \eqref{WeylSpheroidalCordiantes}, the potential equations \eqref{DefinicionPotencialesGeneral} expressed in spheroidal coordinates takes the following form:
{\small
\begin{subequations}\label{FuncionesMetricas-xy}
\begin{center}
    \text{\textbf{Differential equation for $\omega$}}
    \begin{equation}\label{EcuacionesDiferencialesOmega}
        \begin{bmatrix}
            \partial_x \\
            \partial_y \\
        \end{bmatrix} 
        \omega= \frac{L}{f^2}\Big( \begin{bmatrix}
            (1-y^2)\partial_y  \\
            -(x^2\pm1)\partial_x  
        \end{bmatrix}\epsilon -\psi \begin{bmatrix}
            (1-y^2)\partial_y  \\
            -(x^2\pm1)\partial_x  
        \end{bmatrix}\chi \Big) ,
    \end{equation}

    \text{\textbf{Differential equation for $A_\varphi$}}
    \begin{equation}\label{EcuacionesDiferencialesA3}
        \begin{bmatrix}
            \partial_x \\
            \partial_y \\
        \end{bmatrix} 
        A_\varphi = \frac{1}{2  Lf \kappa^2 } \begin{bmatrix}
            (1-y^2)\partial_y  \\
            -(x^2\pm1)\partial_x  
        \end{bmatrix} \chi
        -\frac{\omega}{2L} 
        \begin{bmatrix}
            \partial_x \\
            \partial_y \\
        \end{bmatrix} \psi.
    \end{equation}
    \end{center}
\end{subequations}
}
And the metric function $k$ is given by the following quadrature
{\footnotesize
\begin{subequations}\label{kVariableCompleja}
\begin{align}
    k_{,\zeta}=& \frac{\rho}{4f^2} \bigg[ \tensor{f}{_{,\zeta}}{^2}+(\tensor{\epsilon}{_{,\zeta}} -\psi \tensor{\chi}{_{,\zeta}})^2+f \left( \frac{1}{ \kappa^2} \tensor{\chi}{_{,\zeta}}{^2}+ \kappa^2 \tensor{\psi}{_{,\zeta}}{^2} \right) \nonumber \\
    &+\epsilon_0 \left(  \frac{2f}{\alpha_0\kappa}\right)^2 \tensor{\kappa}{_{,\zeta}}{^2} \bigg],\label{k Chi}
\end{align}
\begin{align}
    k_{,\overline{\zeta}}=& \frac{\rho}{4f^2} \bigg[ \tensor{f}{_{,\overline{\zeta}}}{^2}+(\tensor{\epsilon}{_{,\overline{\zeta}}} -\psi \tensor{\chi}{_{,\overline{\zeta}}})^2+f \left( \frac{1}{ \kappa^2} \tensor{\chi}{_{,\overline{\zeta}}}{^2}+ \kappa^2 \tensor{\psi}{_{,\overline{\zeta}}}{^2} \right) \nonumber \\
    &+\epsilon_0 \left(  \frac{2f}{\alpha_0\kappa}\right)^2 \tensor{\kappa}{_{,\overline{\zeta}}}{^2} \bigg].\label{k complejaconjugada}
\end{align}
\end{subequations}
}
where $\zeta=\rho +i z$.

It is important to note that each partial derivative quadratic individual contribution term $\mathtt{F}_{,\zeta}$ or $\mathtt{F}_{,\overline \zeta}$ with $\mathtt F \in \{f,\epsilon,\psi,\chi,\kappa \}$ takes the following form in spheroidal coordinates:
{\small
\begin{subequations}\label{kCadaTerminoCuadradito}
\begin{align}
    &\partial_x k \propto \frac{1}{4f^2} \frac{(1-y^2)}{x^2\pm y^2} \bigg\{ x\Big[ (x^2\pm1)(\partial_x \mathtt F)^2-(1-y^2) (\partial_y \mathtt F)^2  \Big] \nonumber \\
    & \qquad -2y (x^2\pm1)(\partial_x \mathtt F)(\partial_y \mathtt F) \bigg\}, \label{EcuacionesDiferencialesk1Completo}\\
    &\partial_y k \propto \frac{1}{4f^2} \frac{(x^2\pm 1)}{x^2\pm y^2} \bigg\{ y \Big[ (x^2\pm 1)(\partial_x \mathtt F)^2-(1-y^2) (\partial_y \mathtt F)^2  \Big] \nonumber \\ 
    & \qquad + 2x (1-y^2)(\partial_x \mathtt F)(\partial_y \mathtt F) \bigg\} \label{EcuacionesDiferencialesk2Completo} ,
\end{align}
\end{subequations}
}

\end{appendices}
\bibliographystyle{elsarticle-harv} 
\bibliography{Bibliografia}

\end{document}